\documentclass[a4paper,UKenglish,cleveref, autoref, thm-restate,svgnames,dvipsnames]{lipics-v2021}
\pdfoutput=1

\hideLIPIcs  

\title{On a Characterization of Spartan Graphs} 

\author{Neeldhara Misra}{IIT Gandhinagar, India}{neeldhara.m@iitgn.ac.in }{}{Supported by the SERB Early Career Researcher Grant ECR/2018/002967}

\author{Saraswati Girish Nanoti}{IIT Gandhinagar, India}{nanoti_saraswati@iitgn.ac.in}{}{Supported by CSIR}

\authorrunning{N. Misra, and S. Nanoti} 

\Copyright{} 
    


\category{} 

\relatedversion{} 

\acknowledgements{}

\nolinenumbers 

\EventEditors{John Q. Open and Joan R. Access}
\EventNoEds{2}
\EventLongTitle{42nd Conference on Very Important Topics (CVIT 2016)}
\EventShortTitle{CVIT 2016}
\EventAcronym{CVIT}
\EventYear{2016}
\EventDate{December 24--27, 2016}
\EventLocation{Little Whinging, United Kingdom}
\EventLogo{}
\SeriesVolume{42}
\ArticleNo{23}

\usepackage{charter,eulervm,parskip,tikz,float,nicefrac,framed}
\usepackage{xcolor,graphicx}
\usetikzlibrary{shapes.geometric, calc, positioning, decorations.pathreplacing}
\usetikzlibrary{decorations.shapes,decorations.pathreplacing,decorations.pathmorphing}
\usetikzlibrary{arrows,matrix,shapes}
\usepackage{tcolorbox}
\tikzset{
    small circles/.style={circle,inner sep=2pt,fill=#1},
    hollow circles/.style n args={2}{circle,inner sep=#1,draw=#2,thick},
    stars/.style={star,inner sep=2pt}
}
\definecolor{DarkGreen}{rgb}{0.1,0.5,0.1}
\hypersetup{
colorlinks=true,
linkcolor=IndianRed,
urlcolor=DarkGreen,
citecolor=SeaGreen
}

\begin{document}

\maketitle

\begin{abstract}
The eternal vertex cover game is played between an attacker and a defender on an undirected graph $G$. The defender identifies $k$ vertices to position guards on to begin with. The attacker, on their turn, attacks an edge $e$, and the defender must move a guard along $e$ to defend the attack. The defender may move other guards as well, under the constraint that every guard moves at most once and to a neighboring vertex. The smallest number of guards required to defend attacks forever is called the eternal vertex cover number of $G$, denoted $\mathsf{evc}(G)$.

For any graph $G$, $\mathsf{evc}(G)$ is at least the vertex cover number of $G$, denoted $\mathsf{mvc}(G)$. A graph is Spartan if $\mathsf{evc}(G) = \mathsf{mvc}(G)$. It is known that a bipartite graph is Spartan if and only if every edge belongs to a perfect matching. We show that the only König graphs that are Spartan are the bipartite Spartan graphs. We also give new lower bounds for $\mathsf{evc}(G)$, generalizing a known lower bound based on cut vertices. We finally show a new matching-based characterization of all Spartan graphs.
\end{abstract}
\section{Introduction}\label{sec:introduction}
Recall that a subset of vertices $S$ is called a \emph{vertex cover} if, for every edge $e = \{u,v\}$, at least one of $u$ or $v$ belong to $S$. The eternal vertex cover game, introduced by Klostermeyer and Mynhardt~\cite{KM09}, is a turn-based two player game played on a graph between players typically referred to as the attacker and the defender. The defender is tasked with placing guards on the vertices of the graph to protect against attacks on edges, while the attacker selects edges to attack. The defender's goal is to ensure that every attack can be countered by moving a guard along the attacked edge, while optionally moving additional guards, under the constraint that guards move at most once, and to a neighboring vertex. When the defender is able to defend attacks forever, then the positions of the guards at every stage of the game form a vertex cover. 

The minimum number of guards required for the defender to successfully defend against an infinite sequence of attacks is known as the \emph{eternal vertex cover number}, denoted as $\mathsf{evc}(G)$. We also use $\mathsf{mvc}(G)$ to denote the \emph{minimum} vertex cover number of $G$, which is the size of a smallest vertex cover of $G$. It is well known that~\cite{KM09}:

$$\underbrace{\mathsf{mvc}(G) \leqslant \mathsf{evc}(G)}_{\text{The placement of guards\\ is a vertex cover.}} \text{ and } \underbrace{\mathsf{evc}(G) \leqslant 2\cdot \mathsf{mvc}(G)}_{\text{Guard both endpoints of a maximum matching.}}$$

A natural question is to identify the graphs for which the lower bound is tight, i.e, when $\mathsf{evc}(G) = \mathsf{mvc}(G)$. Such graphs are called \emph{Spartan} --- here, the defender can maintain a vertex cover with the minimum number of guards, making them an optimal solution in terms of resource usage. We are interested in the question of what such graphs look like from a structural perspective. 

In~\cite{BCFPRW2021}, Babu et al achieve a characterization for a class of graphs that includes chordal graphs and internally triangulated planar graphs. We briefly recall this result. Let the graph class $\mathcal{F}$ denote the class of all connected graphs $G$ for which each minimum vertex cover of $G$ that contains all the cut vertices of $G$ induces a connected subgraph in $G$. (A cut vertex is a vertex whose removal disconnects the graph.) Let $G(V,E)$ be a graph that belongs to $\mathcal{F}$, with at least two vertices, and $X\subset V$ be the set of cut vertices of $G$. Then it is shown~\cite{BCFPRW2021} that $G$ is Spartan if and only if for every vertex $v\in V\backslash X$, there exists a minimum vertex cover $S_v$ of $G$ such that $X\cup \{v\}\subset S_v$. It is also known that bipartite graphs are Spartan if and only if every edge belongs to a perfect matching~\cite{MN23}. Our first result involves expanding the scope of this characterization to König graphs, which are graphs where the minimum vertex cover equals the maximum matching, and hence a natural generalization of bipartite graphs. We show that the only König graphs that are Spartan are those that are also bipartite and satisfy the conditions for being Spartan within bipartite graphs:

\begin{restatable}{theorem}{konigspartan}
    A K\"onig graph $G$ is Spartan if and only if it is bipartite and essentially elementary.
\end{restatable}

We also develop a series of increasingly stronger pre-requisites --- i.e, necessary conditions --- for a graph to be Spartan. To begin with, note that for a graph $G$ which has more than one vertex, if $\mathsf{evc}(G)=\mathsf{mvc}(G)$ then every vertex $v$ of $G$ must belong to a minimum sized vertex cover $S_v$. Indeed, if not, then from any configuration, attacking an edge incident on $v$ will result in the attacker winning. 
A result in~\cite{BCFPRW2021} takes this further to show that if $\mathsf{evc}(G)=\mathsf{mvc}(G)$ then each vertex must belong to a minimum sized vertex cover which contains all the cut vertices.

\begin{figure}
    \centering
    \begin{tikzpicture}[scale=0.8]
    \draw[Crimson,thick,rounded corners,fill=LightGray!25] (-0.6,10.5) rectangle (2.6,9.5);
    \draw[Crimson,thick,rounded corners,fill=LightGray!25] (7.4,10.5) rectangle (10.6,9.5);
     \draw[fill=Pink] (0,10) circle (0.25 cm);
         \node at (0,10){\small{$a$}};
     \draw[fill=Pink] (1,10) circle (0.25 cm);
         \node at (1,10){\small{$b$}};
     \draw[fill=Pink] (2,10) circle (0.25 cm);
         \node at (2,10){\small{$c$}};
     \draw[fill=Pink] (8,10) circle (0.25 cm);
         \node at (8,10){\small{$d$}};
     \draw[fill=Pink] (9,10) circle (0.25 cm);
         \node at (9,10){\small{$e$}};
     \draw[fill=Pink] (10,10) circle (0.25 cm);
         \node at (10,10){\small{$f$}};
     \draw[thick,DodgerBlue,fill=LightGray!25] (-1.5,8) circle (0.7 cm);
     \node at (-1.5,8){$C_1$};
     \node at (-0.3,8){$\cdots$};
     \draw[thick,DodgerBlue,fill=LightGray!25] (1,8) circle (0.9 cm);
     \node at (2.3,8){$\cdots$};
     \draw[thick,DodgerBlue,fill=LightGray!25] (3.5,8) circle (0.7 cm);
      \node at (3.5,8){$C_\ell$};

     \draw(1,8.5)--(1,9.75);
     \draw(1,8.5)--(0,9.75);
     \draw(0.7,7.8)--(0.9,8.1);
     \draw(1.3,7.8)--(1.1,8.1);
     
     \draw[dashed,thick] (-1.5,8.7)--(0.5,9.5);
     \draw[dashed,thick] (3.5,8.7)--(1.5,9.5);

     \draw[fill=Pink] (0.6,7.7) circle (0.25 cm);
      \node at (0.6,7.7){\small{$p$}};
     \draw[fill=Pink] (1.4,7.7) circle (0.25 cm);
      \node at (1.4,7.7){\small{$r$}};
     \draw[fill=Purple!60] (1,8.3) circle (0.25 cm);
     \node at (1,8.3){\small{$q$}};
    
     \draw[thick,DodgerBlue,fill=LightGray!25] (6.5,8) circle (0.7 cm);
     \node at (6.5,8){$C_1$};
     \node at (7.7,8){$\cdots$};
     \draw[thick,DodgerBlue,fill=LightGray!25] (9,8) circle (0.9 cm);
     \node at (10.3,8){$\cdots$};
     \draw[thick,DodgerBlue,fill=LightGray!25] (11.5,8) circle (0.7 cm);
     \node at (11.5,8){$C_\ell$};

     \draw(9,8.75)--(9,9.75);
     \draw(9,8.75)--(10,9.75);
     
     \draw[dashed,thick] (6.5,8.7)--(8.5,9.5);
     \draw[dashed,thick] (11.5,8.7)--(9.5,9.5);

     \draw(8.7,8.3)--(8.9,8.5);
     \draw(9.3,8.3)--(9.1,8.5);
     \draw(8.6,8.05)--(8.6,7.75);
     \draw(8.75,7.6)--(9.25,7.6);

     \draw[fill=Pink] (8.6,8.2) circle (0.15 cm);
     \node at (8.6,8.2){\tiny{w}};
     \draw[fill=Pink] (9.4,8.2) circle (0.15 cm);
     \node at (9.4,8.2){\tiny{y}};
     \draw[fill=Purple!60] (9,8.6) circle (0.15 cm);
     \node at (9,8.6){\tiny{x}};
     \draw[fill=Purple!60] (8.6,7.6) circle (0.15 cm);
      \node at (8.6,7.6){\tiny{v}};
     \draw[fill=Purple!60] (9.4,7.6) circle (0.15 cm);
      \node at (9.4,7.6){\tiny{u}};

    \end{tikzpicture}
    
    \caption{Examples of bad scenarios.}
    \label{fig:stronglygoodweaklygood}
\end{figure}
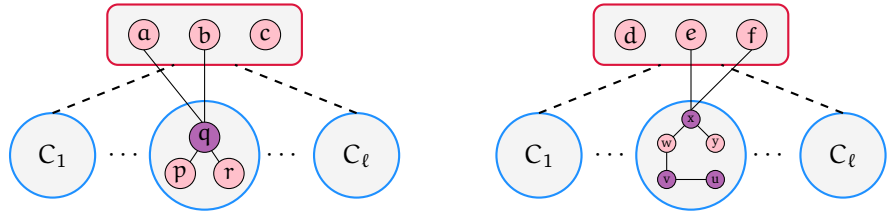


We generalize this idea further in the following way. Let $S$ be a vertex cover of $G$ with independent set $I$. A ``bad situation'' for the defender who has their guards positioned on the vertices of $S$ is the following: suppose there is a $T\subseteq I=V(G)\setminus S$ such that one of the connected components $C$ of $G \setminus T$ is such that $|V(C)\cap S|=\mathsf{mvc}(C)$. This is because the attacker can attack an edge that has one endpoint in $C$ and the other in $T$, forcing a guard out of $C$, and it is easy to check that that the shortfall in $G[C]$ ``cannot be fixed'', creating a vulnerable edge that will lead to an attacker win (see~\Cref{fig:stronglygoodweaklygood} where the defender has one guard on the vertex $q$ with respect to the attack $qb$).

A vertex cover $S$ is \emph{weakly good} if there is no subset $T$ of $V(G) \setminus S$ for which the bad situation described here occurs. We show that a graph $G$ is Spartan only if for each $v\in V(G)$, there exists a minimum sized and weakly good vertex cover $S_v$ such that $v\in S_v$. 

From the defender's perspective, sometimes it is not enough to just avoid the bad situation described above with the vertex covers that they work with. Indeed, suppose we have a component $C$ as shown on the right side of~\Cref{fig:stronglygoodweaklygood}. Here, $|V(C)\cap S| > \mathsf{mvc}(C)$, but note that if the attacker attacks $xe$, then the defender is stuck with guards on $v$ and $u$: while in principle $\mathsf{mvc}(C) = 2$, these two guards cannot be moved in a way that will defend all edges in the component $C$, and as before there is no pathway for guards from outside $C$ to enter $C$. 

While not as stark as before, this is a more nuanced bad scenario. Informally speaking, vertex covers that avoid even this scenario are called \emph{strongly good} vertex covers. We show that a graph is Spartan only if for each $v\in V(G)$, there exists a minimum sized and strongly good vertex cover $S_v$ such that $v\in S_v$.

Our final and main result is a characterization of Spartan graphs at large. Let us begin by making a defender's strategy structurally explicit. To this end, consider a graph where we are able to find a non-empty family $\mathcal{F}$ of minimum-sized vertex covers such that the following holds: 

\emph{For every $S \in \mathcal{F}$ and for every edge $uv \in G$ for which $u \in S$ and $v \notin S$, there exists $T \in \mathcal{F}$ such that and $v \in T$ such that $T$ is reachable from $S$ via guard movement while defending $uv$: this is to say that if the defender had guards positioned on $S$, then they will be able to move them in such a way that the guards finally end up on $T$ and one of the guards moves along the edge $uv$.}

Notice that such a family $\mathcal{F}$ naturally translates to a defender strategy: the defender can start with guards positioned on $S$ for any $S \in \mathcal{F}$, and if any edge $uv$ is attacked, the defender can either exchange guards (if both $u$ and $v$ are in $S$) or move guards to the $T$ guaranteed by the property above. This would work indefinitely since every vertex cover in $\mathcal{F}$ offers the same opportunities to the defender. 

This brings us to the question of capturing movements between minimum-sized vertex covers $S$ and $T$. We refer to the vertices in $V \setminus (S \cup T)$ as the \emph{dead zone}. The defender hopes to move guards currently positioned at $S$ to a new target vertex cover $T$ while also defending an arbitrary but fixed edge $uv$ with (say) $u \in S \setminus T$ and $v \in T \setminus S$. Notice that it suffices to have vertex disjoint paths between $(S \setminus T) \setminus \{u\}$ and $(T \setminus S) \setminus \{v\}$ that do not use any vertices from the dead zone. We say that $S$ and $T$ are \emph{mutually reachable} if such paths exist. 

\begin{figure}
    \centering
    \begin{tikzpicture}[scale=0.75]
    \tikzset{decoration={snake,amplitude=.4mm,segment length=2mm, post length=0mm,pre length=0mm}}
    
     \draw[fill=Pink] (0,10) circle (0.3 cm);
         \node at (0,10){$a$};
     \draw[fill=Pink] (2,10) circle (0.3 cm);
         \node at (2,10){$b$};
     \draw[fill=Pink] (0,9) circle (0.3 cm);
         \node at (0,9){$c$};
     \draw[fill=Pink] (2,9) circle (0.3 cm);
         \node at (2,9){$d$};
     \draw[fill=Pink] (0,8) circle (0.3 cm);
         \node at (0,8){$e$};
     \draw[fill=Pink] (2,8) circle (0.3 cm);
         \node at (2,8){$f$};  
     \draw[fill=Pink] (-1,7) circle (0.3 cm);
         \node at (-1,7){$g$};
     \draw[fill=Pink] (0,6) circle (0.3 cm);
         \node at (0,6){$h$}; 
     \draw[fill=Pink] (1,7) circle (0.3 cm);
         \node at (1,7){$i$};
     \draw[fill=Pink] (3,7) circle (0.3 cm);
         \node at (3,7){$j$}; 
    \draw[fill=Pink] (3,6) circle (0.3 cm);
         \node at (3,6){$k$}; 
    \draw(0.3,10)--(1.7,10);
    \draw(0.3,9)--(1.7,9);
    \draw(0.3,9)--(1.7,10);
    \draw(0.3,8)--(1.7,9);
    \draw(0.3,9)--(1.7,8);
    \draw(-0.7,7)--(0.7,7);
    \draw(-0.8,6.8)--(-0.2,6.2);
    \draw(0.8,6.8)--(0.2,6.2);
    \draw(3,6.7)--(3,6.3);
    \draw(1.2,6.8)--(2.8,6.2);
    \draw(-0.2,9.8)--(-1,7.3);
    \draw(1.2,7.2)--(1.8,7.8);
    \draw(0.3,8)--(2.8,7.2);
    \draw(3,7.3)--(2.2,8.8);
    \draw(0.3,6)--(2.7,6);
     \draw[fill=Pink] (8,10) circle (0.3 cm);
         \node at (8,10){$b$};
     \draw[fill=Pink] (8,9) circle (0.3 cm);
         \node at (8,9){$d$};
     \draw[fill=Pink] (8,8) circle (0.3 cm);
         \node at (8,8){$f$};  
     \draw[fill=Pink] (9,7) circle (0.3 cm);
         \node at (9,7){$j$}; 
    \draw[fill=Pink] (9,6) circle (0.3 cm);
         \node at (9,6){$k$}; 
    \draw(6.3,10)--(7.7,10);
    \draw(6.3,9)--(7.7,9);
    \draw[red,dashed](6.3,9)--(7.7,10);
    \draw[->](6.9,9.4)--(7.1,9.6);
    \draw(6.3,8)--(7.7,9);
    \draw(6.3,9)--(7.7,8);
    \draw(5.3,7)--(6.7,7);
    \draw(5.2,6.8)--(5.8,6.2);
    \draw[->](5.3,6.7)--(5.5,6.5);
    \draw(6.8,6.8)--(6.2,6.2);
    \draw[->](6.4,6.4)--(6.6,6.6);
    \draw(9,6.7)--(9,6.3);
    \draw(7.2,6.8)--(8.8,6.2);
    \draw(5.8,9.8)--(5,7.3);
    \draw[->](5.8,9.8)--(5.4,8.6);
    \draw(7.2,7.2)--(7.8,7.8);
    \draw[->](7.3,7.3)--(7.5,7.5);
    \draw(6.3,8)--(8.8,7.2);
    \draw[->](6.3,8)--(7,7.8);
    \draw(9,7.3)--(8.2,8.8);
    \draw[->](9,7.3)--(8.4,8.4);
    \draw(6.3,6)--(8.7,6);

\node at (6,10)[draw,fill=Purple!70,minimum size=12pt](){$a$};
\node at (6,9)[draw,fill=Purple!70,minimum size=12pt](){$c$};
\node at (6,8)[draw,fill=Purple!70,minimum size=12pt](){$e$};
\node at (5,7)[draw,fill=Purple!70,minimum size=12pt](){$g$};
\node at (6,6)[draw,fill=Purple!70,minimum size=12pt](){$h$};
\node at (7,7)[draw,fill=Purple!70,minimum size=12pt](){$i$};
\node at (9,7)[draw,fill=Purple!70,minimum size=12pt](){$j$};
      \draw[fill=Pink] (12,10) circle (0.3 cm);
      \node at (12,10){$a$};
         
     \draw[fill=Pink] (12,9) circle (0.3 cm);
         \node at (12,9){$c$};
     \draw[fill=Pink] (14,9) circle (0.3 cm);
         \node at (14,9){$d$};
     \draw[fill=Pink] (12,8) circle (0.3 cm);
         \node at (12,8){$e$};
     \draw[fill=Pink] (14,8) circle (0.3 cm);
         \node at (14,8){$f$};  
     \draw[fill=Pink] (11,7) circle (0.3 cm);
         \node at (11,7){$g$};
     \draw[fill=Pink] (12,6) circle (0.3 cm);
         \node at (12,6){$h$}; 
     \draw[fill=Pink] (13,7) circle (0.3 cm);
         \node at (13,7){$i$};
     \draw[fill=Pink] (15,7) circle (0.3 cm);
         \node at (15,7){$j$}; 
    \draw[fill=Pink] (15,6) circle (0.3 cm);
         \node at (15,6){$k$}; 
    \draw(12.3,10)--(13.7,10);
    \draw(12.3,9)--(13.7,9);
    \draw[SeaGreen,thick,decorate](12.3,9)--(13.7,10);
    \draw(12.3,8)--(13.7,9);
    \draw(12.3,9)--(13.7,8);
    \draw(11.3,7)--(12.7,7);
    \draw[red,thick,dashed](11.2,6.8)--(11.8,6.2);
    \draw[red,thick,dashed](12.8,6.8)--(12.2,6.2);
    \draw(15,6.7)--(15,6.3);
    \draw(13.2,6.8)--(14.8,6.2);
    \draw[red,thick,dashed](11.8,9.8)--(11,7.3);
    \draw[red,thick,dashed](13.2,7.2)--(13.8,7.8);
    \draw[Blue,thick,dash dot](12.3,8)--(14.8,7.2);
    \draw[Blue,thick,dash dot](15,7.3)--(14.2,8.8);
    \draw(12.3,6)--(14.7,6);

\node at (14,10)[draw,fill=Purple!70,minimum size=12pt](){$b$};
\node at (14,9)[draw,fill=Purple!70,minimum size=12pt](){$d$};
\node at (14,8)[draw,fill=Purple!70,minimum size=12pt](){$f$};
\node at (11,7)[draw,fill=Purple!70,minimum size=12pt](){$g$};
\node at (12,6)[draw,fill=Purple!70,minimum size=12pt](){$h$};
\node at (13,7)[draw,fill=Purple!70,minimum size=12pt](){$i$};
\node at (15,7)[draw,fill=Purple!70,minimum size=12pt](){$j$};
    \end{tikzpicture}
    \caption{Demonstrating reachability between vertex covers $S = \{a,c,e,g,h,i,j\}$ and $T = \{b,d,f,g,h,i,j\}$ with the vertex $\{k\}$ in the dead zone.}
    \label{fig:pathsfromvcs}
\end{figure}
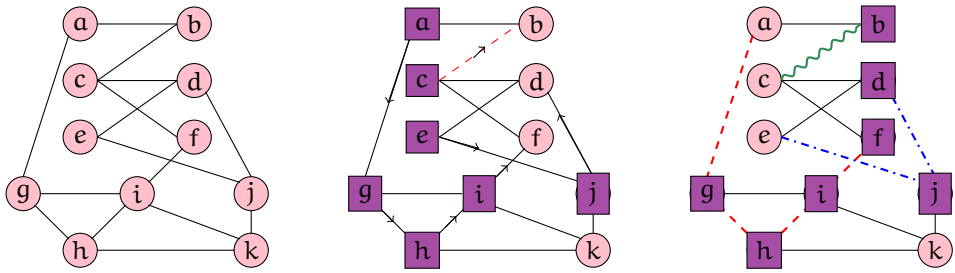

It is intuitive that the existence of a family of mutually reachable minimum vertex covers is a necessary and sufficient condition for a graph to be Spartan: indeed, given such a family, the defender has a strategy, and if the defender has a strategy, then the guard positions naturally correspond to such a family. Our main contribution is to capture the requirement of mutual reachability between vertex covers $S$ and $T$ in terms of matchings in an appropriately defined auxiliary graph $\mathfrak{h}_G(S,T)$ based on $S$ and $T$. The notion of the auxiliary graph required to translate the condition mutual reachability in terms of matchings is somewhat technical and we omit the specifics from the introduction. The non-trivial technical aspect is to ensure that the reachability guaranteed by disjoint paths as described above is in fact captured by the existence of a matching in $\mathfrak{h}_G(S,T)$.

\section{Preliminaries}\label{sec:prelims}

Let $G(V,E)$ be a simple, finite and undirected graph with $n>1$ vertices and $m$ edges (unless mentioned otherwise). Also unless mentioned otherwise, we will assume that the graph $G$ is connected (and in particular, has no isolated vertices). We use standard notation from~\cite{graphsbook}. If $S$ is a subset of $V$, $G[S]$ denotes the subgraph of $G$ induced by the vertices in $S$. The set of all vertices $v$ such that $uv\in E(G)$ is denoted by $N(u)$ and is said to be the \emph{open neighbourhood} of $u$. The set $N(u)\cup\{u\}$ is said to be the \emph{closed neighbourhood} of $v$.

A \textit{path} on $n$ vertices is a graph $P$ with vertices $\{v_1,v_2,\ldots, v_n\}$ such that each $(v_i,v_{i+1})\in E(P)$ for $i\in[1,n-1]$. A \textit{cycle} on $n$ vertices is a graph $C$ with vertices $\{v_1,v_2,\ldots ,v_n\}$ such that each $(v_i,v_{i+1})\in E(C)$ for $i\in[1,n-1]$ and $(v_n,v_1)\in E(C)$. A graph $G(V,E)$ is said to be \textit{connected} if for every two distinct vertices $u,v\in V$, there exists a path between $u$ and $v$. 

A subset $S$ of $V(G)$ is said to be a \textit{vertex cover} of $G$ if for every edge $(u,v)\in E$, either $u\in S$ or $v\in S$. The size of the smallest vertex cover of the graph $G$ is called the \textit{minimum vertex cover number} of $G$ and denoted by $\mathsf{mvc}(G)$. The graphs with $\mathsf{evc}(G)=\mathsf{mvc}(G)$ are known as \emph{Spartan graphs}.

A subset $S$ of $V(G)$ is said to be an \textit{independent set} of $G$ if no edge of $E(G)$ has both its endpoints in $S$. It can be seen that the complement of a vertex cover will be an independent set and vice versa. 

A \emph{matching} is a set of edges with no common endpoint. A \emph{perfect matching} is matching which contains one edge adjacent to each vertex of a graph. A graph may not always have a perfect matching. A matching of the largest cardinality in a graph is called a \emph{maximum matching} and the size of this matching is denoted by $mm(G)$. Examples are shown in \Cref{fig:matching}, \Cref{fig:perfectmatching} and \Cref{fig:maxmatching}.

A \emph{bipartite graph} $G = (V,E)$ is a graph whose vertex set can be partitioned into two independent sets, say $V = (A \cup B)$, that is every edge is between a vertex in $A$ and one in $B$. Clearly, both $A$ and $B$ are vertex covers of $G$. If a bipartite graph $G = (A \cup B, E)$ is connected and its \emph{only} optimal vertex covers are $A$ and $B$, then we say that $G$ is \emph{elementary}. If every connected component of a bipartite graph is elementary, then we call it \emph{essentially elementary}.

In \Cref{sec:konig}, we need to use the following result about bipartite graphs which is given in \cite{LLPM-matchingbook} and can also be obtained by a slight modification of a result in \cite{MN23}.

\begin{lemma}\label{degone}
    Let $G$ be a connected bipartite graph such that $V(G)=A\cup B$ where $A$ and $B$ are non-empty independent sets. If $G$ has at least one perfect matching and $G$ is non-elementary, then there exist non-empty $S\subsetneq A$ and $T\subsetneq B$ such that $|N(S)|=|S|$ and $|N(T)|=|T|$.
\end{lemma}
\begin{proof}
  Let $G$ be a non-elementary connected bipartite graph with a $V(G)=A\cup B$ where $A$ and $B$ are non-empty independent sets and a perfect matching. Suppose that for every non-empty subset $S$ of $A$, $S\subsetneq A$ we have $|N(S)|>|S|$. Let $ab\in E(G)$. We show that there exists a perfect matching between $A\setminus \{a\}$ and $B\setminus\{b\}$. Consider any non-empty $X\subseteq A\setminus \{a\}$. If $|N(X)\cap(B\setminus\{b\})|<|X|$, then $|N(X)\cap B|\leq |X|$. That is, $X\subsetneq A$ such that $|N(X)|\leq |X|$ which is contrary to our assumption that for every subset $S$ of $A$, $S\subsetneq A$ we have $|N(S)|>|S|$. Thus we cannot have any $X\subseteq A\setminus \{a\}$ such that $|N(X)\cap(B\setminus\{b\})|<|X|$. Therefore, by Hall's theorem, there exists a perfect matching $M_1$ between $A\setminus \{a\}$ and $B\setminus \{b\}$. Thus $M=M_1\cup \{ab\}$ is a perfect matching of $G$ containing $ab$. Since $ab$ was an arbitrary edge of $G$, every edge of $G$ belongs to some perfect matching and $G$ is connected which implies that $G$ is elementary which is a contradiction. 
  Therefore, there must exist a non-empty subset $S$ of $A$ such that $S\subsetneq A$ and $|N(S)|\leq |S|$ but $|N(S)|$ cannot be less than $|S|$ because $G$ has a perfect matching. Thus $|N(S)|=|S|$. With a symmetric argument, we get that there exists a non-empty $T\subsetneq B$ such that $|N(T)|=|T|$.
 \end{proof}

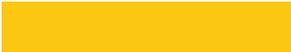
\begin{figure}
    \centering
    \scalebox{1.3}{
\begin{tikzpicture}


\node[small circles=DodgerBlue] (A1) at (0,0) {};
\node[small circles=DodgerBlue] (B1) [below of=A1] {};

\foreach \x in {2,3,...,6}{
    \pgfmathtruncatemacro{\y}{\x-1}
    \node[small circles=DodgerBlue] (A\x) [right of=A\y] {};
}
\foreach \x in {2,3,...,5}{
    \pgfmathtruncatemacro{\y}{\x-1}
    \node[small circles=DodgerBlue] (B\x) [right of=B\y] {};
};
\node (B6) [right of = B5] {};    
\node[small circles=DodgerBlue] (B7) [right of = B6] {};

\draw (A1) -- (B1) -- (A2) -- (B3) -- (A3)  -- (B4)  -- (A4) -- (B5) -- (A6) -- (B7);
\draw (A1) -- (B2);
\draw (B1) -- (A2);
\draw (A1) -- (B2);
\draw (A5) -- (B7);
\draw (A5) -- (B4);
\draw (A3) -- (B1);

\draw[thick,color=SeaGreen,decoration = {zigzag,segment length = 3pt,amplitude=1pt},decorate] (A1) -- (B1);

\draw[thick,color=SeaGreen,decoration = {zigzag,segment length = 3pt,amplitude=1pt},decorate] (A3) -- (B3);

\draw[thick,color=SeaGreen,decoration = {zigzag,segment length = 3pt,amplitude=1pt},decorate] (A4) -- (B4);

\end{tikzpicture}
}
   
\caption{The green wavy edges form a matching in the given graph $G$}
    \label{fig:matching}
\end{figure}
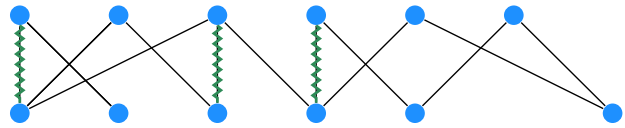
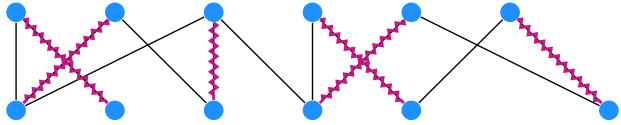
\begin{figure}
    \centering
   \scalebox{1.3}{
\begin{tikzpicture}


\node[small circles=DodgerBlue] (A1) at (0,0) {};
\node[small circles=DodgerBlue] (B1) [below of=A1] {};

\foreach \x in {2,3,...,6}{
    \pgfmathtruncatemacro{\y}{\x-1}
    \node[small circles=DodgerBlue] (A\x) [right of=A\y] {};
}
\foreach \x in {2,3,...,5}{
    \pgfmathtruncatemacro{\y}{\x-1}
    \node[small circles=DodgerBlue] (B\x) [right of=B\y] {};
};
\node (B6) [right of = B5] {};    
\node[small circles=DodgerBlue] (B7) [right of = B6] {};

\draw (A1) -- (B1) -- (A2) -- (B3) -- (A3)  -- (B4)  -- (A4) -- (B5) -- (A6) -- (B7);
\draw (A1) -- (B2);
\draw (B1) -- (A2);
\draw (A1) -- (B2);
\draw (A5) -- (B7);
\draw (A5) -- (B4);
\draw (A3) -- (B1);

\draw[thick,color=MediumVioletRed,decoration = {zigzag,segment length = 3pt,amplitude=1pt},decorate] (A1) -- (B2);

\draw[thick,color=MediumVioletRed,decoration = {zigzag,segment length = 3pt,amplitude=1pt},decorate] (A2) -- (B1);

\draw[thick,color=MediumVioletRed,decoration = {zigzag,segment length = 3pt,amplitude=1pt},decorate] (A3) -- (B3);

\draw[thick,color=MediumVioletRed,decoration = {zigzag,segment length = 3pt,amplitude=1pt},decorate] (A4) -- (B5);

\draw[thick,color=MediumVioletRed,decoration = {zigzag,segment length = 3pt,amplitude=1pt},decorate] (A5) -- (B4);

\draw[thick,color=MediumVioletRed,decoration = {zigzag,segment length = 3pt,amplitude=1pt},decorate] (A6) -- (B7);

\end{tikzpicture}
}
\label{fig:perfectmatching}
\caption{The purple wavy edges form a perfect matching in the given graph $G$}
\end{figure}
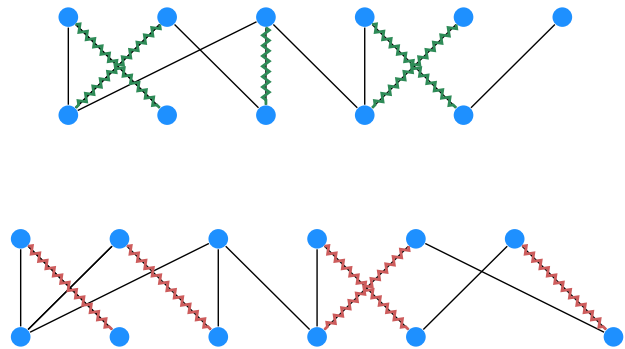
\begin{figure}
\centering
\scalebox{1.3}{
\begin{tikzpicture}


\node[small circles=DodgerBlue] (A1) at (0,0) {};
\node[small circles=DodgerBlue] (B1) [below of=A1] {};

\foreach \x in {2,3,...,6}{
    \pgfmathtruncatemacro{\y}{\x-1}
    \node[small circles=DodgerBlue] (A\x) [right of=A\y] {};
}
\foreach \x in {2,3,...,5}{
    \pgfmathtruncatemacro{\y}{\x-1}
    \node[small circles=DodgerBlue] (B\x) [right of=B\y] {};
};
\node (B6) [right of = B5] {};

\draw (A1) -- (B1) -- (A2) -- (B3) -- (A3)  -- (B4)  -- (A4) -- (B5) -- (A6);
\draw (A1) -- (B2);
\draw (B1) -- (A2);
\draw (A1) -- (B2);

\draw (A5) -- (B4);
\draw (A3) -- (B1);

\draw[thick,color=SeaGreen,decoration = {zigzag,segment length = 3pt,amplitude=1pt},decorate] (A1) -- (B2);

\draw[thick,color=SeaGreen,decoration = {zigzag,segment length = 3pt,amplitude=1pt},decorate] (A2) -- (B1);

\draw[thick,color=SeaGreen,decoration = {zigzag,segment length = 3pt,amplitude=1pt},decorate] (A3) -- (B3);

\draw[thick,color=SeaGreen,decoration = {zigzag,segment length = 3pt,amplitude=1pt},decorate] (A4) -- (B5);

\draw[thick,color=SeaGreen,decoration = {zigzag,segment length = 3pt,amplitude=1pt},decorate] (A5) -- (B4);

\end{tikzpicture}
}

\scalebox{1.3}{
\begin{tikzpicture}

\node at (0,0){};
\node[small circles=DodgerBlue] (A1) at (0,-1) {};
\node[small circles=DodgerBlue] (B1) [below of=A1] {};

\foreach \x in {2,3,...,6}{
    \pgfmathtruncatemacro{\y}{\x-1}
    \node[small circles=DodgerBlue] (A\x) [right of=A\y] {};
}
\foreach \x in {2,3,...,5}{
    \pgfmathtruncatemacro{\y}{\x-1}
    \node[small circles=DodgerBlue] (B\x) [right of=B\y] {};
};
\node (B6) [right of = B5] {};    
\node[small circles=DodgerBlue] (B7) [right of = B6] {};

\draw (A1) -- (B1) -- (A2) -- (B3) -- (A3)  -- (B4)  -- (A4) -- (B5) -- (A6) -- (B7);
\draw (A1) -- (B2);
\draw (B1) -- (A2);
\draw (A1) -- (B2);
\draw (A5) -- (B7);
\draw (A5) -- (B4);
\draw (A3) -- (B1);

\draw[thick,color=IndianRed,decoration = {zigzag,segment length = 3pt,amplitude=1pt},decorate] (A1) -- (B2);

\draw[thick,color=IndianRed,decoration = {zigzag,segment length = 3pt,amplitude=1pt},decorate] (A2) -- (B3);

\draw[thick,color=IndianRed,decoration = {zigzag,segment length = 3pt,amplitude=1pt},decorate] (A4) -- (B5);

\draw[thick,color=IndianRed,decoration = {zigzag,segment length = 3pt,amplitude=1pt},decorate] (A5) -- (B4);

\draw[thick,color=IndianRed,decoration = {zigzag,segment length = 3pt,amplitude=1pt},decorate] (A6) -- (B7);

\end{tikzpicture}
}

\caption{The green wavy edges in the top figure do not form a maximum matching of the given graph $G$ although the matching is maximal, i.e., no other edge of the graph can be added to it to form a bigger matching. The red wavy edges in the bottom figure form a maximum matching of the given graph $G$.}
\label{fig:maxmatching}
\end{figure}

\newpage
\section{New Lower Bounds}\label{sec:lowerbounds}
In this section we give some necessary conditions which a graph $G$ must satisfy in order to have $\mathsf{evc}(G)=\mathsf{mvc}(G)$. One such condition is stated in \cite{BCFPRW2021} which is as follows:
\begin{lemma}\label{vinVC}
    For a graph $G$ which has more than one vertex, if $\mathsf{evc}(G)=\mathsf{mvc}(G)$ then every vertex $v$ of $G$ must belong to a minimum sized vertex cover $S_v$.
\end{lemma}
\begin{proof}
Consider a graph $G$ which has more than one vertex and $\mathsf{evc}(G)=\mathsf{mvc}(G)$ such that there exists $v\in V(G)$ such that $G$ has no minimum sized vertex cover which contains $v$. Therefore, in the initial configuration formed by the guards $v$ is unoccupied. Let $u$ be a neighbour of $v$. Now suppose the attacker attacks the edge $uv$. If a guard was not present on $u$, the attacker wins. If a guard was present on $u$ before the attack, the guard on $u$ moves to $v$ and since there is no minimum sized vertex cover of $G$ which contains $v$, the remaining guards cannot arrange themselves to form a vertex cover of $G$. Thus some edge must remain vulnerable after this attack and thus the attacker can attack it in the next round and win.
 
\end{proof}

For understanding the nature of Spartan graphs, it is sufficient to look at only connected components. The lemma below which was proven in \cite{MN23} justifies this.
\begin{lemma}\label{connected}
    Now, let $G$ be a Spartan graph with connected components $C_1, \ldots, C_\ell$. If $G$ is Spartan then $G[C_i]$ is Spartan for all $1 \leqslant i \leqslant [\ell]$.
\end{lemma}
\begin{proof}
    Note that: 

$$\mathsf{evc}(G) = \sum_{i=1}^{\ell}\mathsf{evc}(G[C_i]) \geqslant \sum_{i=1}^{\ell}\mathsf{mvc}(G[C_i]) = \mathsf{mvc}(G).$$


Therefore, if $G$ is Spartan then $G[C_i]$ is Spartan for all $1 \leqslant i \leqslant [\ell]$. Indeed, if not, then there exists a component $C_i$ for which $\mathsf{evc}(G[C_i]) > \mathsf{mvc}(G[C_i])$. But combined with the inequality above, this will imply that $\mathsf{evc}(G) > \mathsf{mvc}(G)$, contradicting our assumption that $G$ is Spartan. 
\end{proof}
Thus we will now look at only connected graphs.
\begin{lemma}\label{matchI}
    If a graph $G$ with more than one vertex is Spartan and $I$ is a maximum independent set of $G$, then there is a matching of $I$ to $V(G)\setminus I$ which saturates $I$.
\end{lemma}
\begin{proof}
    Consider a Spartan graph $G$ (with more than one vertex) and let $I$ be a maximum independent set of $G$ and $C=V(G)\setminus I$ be a minimum sized vertex cover. Suppose there is no matching from $I$ to $S$ which saturates $I$, then by Hall's theorem there exists $X\subseteq I$ such that $|N(X)|<|X|$. (Here, the neighborhood of any subset of $I$ in the entire graph is the same as the neighborhood of that subset in $C$ because $I$ is an independent set.)
    
 Consider an inclusion-wise minimal such subset $X$ of $I$. Then for every $x\in X$, there exists a perfect matching from $X\setminus\{x\}$ to $N(X)$. If not, then there exists a $Y\subseteq X\setminus\{x\}$ such that $|N(Y)|<|Y|$ and as $Y\subseteq X\setminus \{x\}$ which means that $Y\subset X$, this contradicts the inclusion-wise minimality of $Y$. Hence there exists a perfect matching from $X\setminus \{x\}$ to $N(X)$ and this also means that $|X|=|N(X)|+1$. Next we will prove a claim which demonstrates how any vertex cover of $G$ intersects the vertices in $X\cup N(X)$. 

\begin{claim}\label{cl:vfromN(X)}
        Any minimum sized vertex cover of $G$ cannot contain more than $|N(X)|$ many vertices from $X\cup N(X)$.
 \end{claim}
\begin{proof}
Let $T$ be a minimum sized vertex cover of $G$ such that $T$ contains more than $|N(X)|$ many vertices from $X\cup N(X)$. Consider the set $T^{\prime}=\big(T\setminus(X\cup N(X))\big)\cup N(X)$. Clearly, the size of $T^{\prime}$ is less than the size of $T$. We show that $T^{\prime}$ is also a vertex cover of $G$ which contradicts the fact that $T$ is a minimum sized vertex cover of $G$. Any edge with both endpoints outside $X\cup N(X)$ is covered by $T$ and hence covered by $T^\prime$ because the vertices outside $X\cup N(X)$ which belong to $T$ also belong to $T^\prime$. Any edge with both endpoints in $X\cup N(X)$ is also covered by $T^\prime$ as $N(X)\subset T^\prime$ and no edge in $G$ can have both endpoints in $X$ as $X$ is an independent set. An edge with one endpoint in $X\cup N(X)$ and other endpoint outside $X\cup N(X)$ must have an endpoint in $N(X)$ because a vertex in $X$ cannot be adjacent to a vertex outside $N(X)$ (by the definition of $N(X)$). Thus such an edge will also be covered by $T^\prime$ as $N(X)\subseteq T^\prime$. Thus every edge of $G$ has at least one endpoint in $T^\prime$. This is a contradiction as $G$ cannot have a vertex cover of size smaller than $T$ and thus any minimum sized vertex cover of $G$ cannot contain more than $N(X)$ vertices from $X\cup N(X)$. \end{proof}

Now consider a vertex $x\in X$. Since $G$ is Spartan, there is a minimum sized vertex cover $S_x$ of $G$ such that $x\in S_x$ by \Cref{vinVC}. We have already shown that $X\setminus\{x\}$ has a perfect matching with $N(X)$. Thus any vertex cover must contain $|N(X)|$ many vertices from $X\setminus \{x\}\cup N(X)$. Thus $S_x$ contains $x$ and $|N(X)|$ many vertices from $X\setminus \{x\}\cup N(X)$. That is, $S_x$ contains $|N(X)|$ vertices from $X\cup N(X)$. This contradicts the above claim. Thus the existence of a Hall's violator set $X$ is not possible and there exists a matching from $I$ to $S=V(G)\setminus I$ which saturates $I$.
\end{proof}
This immediately gives us a simple corollary. 
\begin{corollary}
   If $G$ is a Spartan graph with more than one vertex, then $\mathsf{mvc}(G)\geq \nicefrac{n}{2}$. 
\end{corollary}
\begin{proof}
    Let $G$ be a Spartan graph with more than one vertex and $I$ be a maximum independent set of $G$. Let $S=V(G)\setminus I$. Then $|N(I)|\geq|I|$ by \Cref{matchI} and thus $|I|\leq \nicefrac{|V(G)|}{2}$ which means that $|S|\geq \nicefrac{|V(G)|}{2}$ i.e. $\mathsf{mvc}(G)\geq \nicefrac{|V(G)|}{2}$.
\end{proof}

\begin{definition}
    A vertex cover $S$ of a graph $G$ is said to be a \emph{weakly good} vertex cover if it satisfies the following:
    Let $I=V(G)\setminus S$ be the corresponding independent set. For every $T\subseteq I$, if $C_1,C_2,\ldots C_\ell$ are the connected components of $G[V(G)\setminus T]$, we must have $|V(C_i)\cap S|>\mathsf{mvc}(C_i)$ (for all $i\in [1,\ell]$). 
\end{definition}

When we are dealing with more than $\mathsf{mvc}(G)$ many guards and we are in the model of the game where more than one guard is able to occupy a single vertex, then we will need to look at configurations of guards, not just vertex covers or vertex sets. A configuration of guards is a description of how many guards are present on each vertex of the graph $G$. A configuration can also be viewed as a function which maps each vertex to a non-negative integer such that the sum of the values of this function on all the vertices of $G$ is equal to the number of guards. At some places, we will also view a configuration as a multi-set where a vertex appears as many times as the number of guards present on it. The manner in which we are using the word configuration will be clear from context. When we say a configuration of guards $\mathcal{C}$ on a vertex cover $S$, we mean that the vertices having one or more guards in $\mathcal{C}$ (i.e. a non-zero value for the configuration function) are precisely the ones in $S$. 

\begin{definition}
    A configuration of guards $\mathcal{C}$ on a vertex cover $S$ is said to be a \emph{weakly good} configuration if it satisfies the following:
    Let $I=V(G)\setminus S$ be the corresponding independent set. For every $T\subseteq I$, if $C_1,C_2,\ldots C_\ell$ are the connected components of $G[V(G)\setminus T]$, the number of guards in $C_i$ must be greater than $\mathsf{mvc}(C_i)$ (for all $i\in [1,\ell]$). 
\end{definition}

If the number of guards on $G$ is the same as $\mathsf{mvc}(G)$, or we are working in the model where only one guard per vertex is allowed, then the vertex cover formed by vertices occupied by the guards in a weakly good configuration is a weakly good vertex cover. It is easier to see the notion of a weakly good vertex cover (or configuration) by looking at a vertex cover (or configuration) which is not weakly good. 

\begin{definition}
Let $S$ be a vertex cover of a graph $G$ such that $S$ is not weakly good. This means that there exists a non-empty $T\subseteq I=V(G)\setminus S$ such that if $C_1,C_2,\ldots C_\ell$ are the connected components of $G[V(G)\setminus T]$, we must have $|V(C_i)\cap S|=\mathsf{mvc}(C_i)$ for some $C_i$ (where $i\in [1,\ell]$). Then we will denote $T$ as a \emph{weakly bad set} corresponding to $S$.
\end{definition}

\begin{definition}
Let $S$ be a vertex cover of a graph $G$ and let $\mathcal{C}$ be a configuration of guards on $S$ such that $\mathcal{C}$ is not weakly good. This means that there exists a non-empty $T\subseteq I=V(G)\setminus S$ such that if $C_1,C_2,\ldots C_\ell$ are the connected components of $G[V(G)\setminus T]$, there must exist a $C_i$ (where $i\in [1,\ell]$) such that the number of guards in $C_i$ must be equal to $\mathsf{mvc}(C_i)$ (cannot be less than $\mathsf{mvc}(C_i)$ because $S$ is a vertex cover). Then we will denote $T$ as a \emph{weakly bad set} corresponding to $S$.
\end{definition}

For every vertex cover $S$ which is not weakly good, there must exist at least one corresponding weakly bad set $T\subset V(G)\setminus S$. Similarly for every configuration $\mathcal{C}$ which is not weakly good, there must exist at least one corresponding weakly bad set $T\subset V(G)\setminus S$. Conversely, if for a vertex cover $S$ (or a configuration $\mathcal{C}$ occupying the vertices in a vertex cover $S$), there exists a subset $T$ of $V(G)\setminus S$ which is weakly bad, then the vertex cover $S$ (or the configuration $\mathcal{C}$) is not weakly good.

\begin{lemma}\label{wgoodVC}
A graph $G$ has $\mathsf{evc}(G)=k$ only if for each $v\in V(G)$, there exists a weakly good configuration $\mathcal{C}_v$ corresponding to a vertex cover $S_v$ such that $\mathcal{C}_v$ has $k$ guards and $v\in S_v$. In particular, a graph $G$ is Spartan only if for each $v\in V(G)$, there exists a minimum sized and weakly good vertex cover $S_v$ such that $v\in S_v$.
\end{lemma}
\begin{proof}
Suppose a graph $G$ has $\mathsf{evc}(G)=k$ and there exists a vertex $v$ such that $v\in V(G)$ and there does not exist a weakly good configuration with $k$ guards with a guard on $v$.  Without loss of generality, we can assume that the defender starts with a configuration with a guard on $v$ because the attacker can always attack an edge adjacent to $v$ and ensure that at least one guard comes to $v$. As there exists no weakly good configuration with $k$ guards with a guard on $v$, tthe configuration $\mathcal{C}$ formed by the guards cannot be a weakly good configuration. If the guards do not occupy a vertex cover, some edge must be vulnerable and the attacker can attack that edge. So we can assume that the set $S$ of vertices occupied by the guards the configuration $\mathcal{C}$ forms a vertex cover. Therefore there exists a weakly bad subset $T$ of the unoccupied vertices such that if $C_1,C_2,\ldots C_\ell$ are the connected components of $G[V(G)\setminus T]$, we must have a $C_i$ (where $i\in [1,\ell]$ which contains $\mathsf{mvc}(C_i)$ many guards. Now since the components are not adjacent to each other, there exists an edge from a vertex $w\in V(C_i)$ to a vertex $u\in T$ because the graph $G$ is connected. Since no vertex in $T$ has a guard, the guard on $w$ must come to $u$. Therefore some edge in $C_i$ will become vulnerable as there were only $\mathsf{mvc}(C_i)$ many guards in $V(C_i)$ and one guard has moved out of $C_i$ while no guard from outside can come to a vertex of $C_i$ in one step. Thus the attacker wins and this contradicts $\mathsf{evc}(G)=k$.
\end{proof}
\begin{lemma}\label{cut}
    Let $S$ be a minimum sized vertex cover of a graph $G$ such that there exists a cut vertex $x$ of $G$ such that $x\notin S$. Then $S$ cannot be a weakly good vertex cover.
\end{lemma}
\begin{proof}
Let $x$ be a cut vertex of a graph $G$ and let $S$ be a vertex cover of $G$ which does not contain $x$. Then we show that $\{x\}\subset V(G)\setminus S$ is weakly bad which will imply that $S$ is not a weakly good vertex cover. Since $x$ is a cut vertex, there exist connected components $C_1, C_2,\ldots, C_{\ell}$ where $\ell\geq 2$ of $V(G)\setminus \{x\}$. If $|S\cap V(C_i)|>\mathsf{mvc}(C_i)$ for all $i\in [1,\ell]$, then $S^\prime=\{x\}\cup S_1\cup S_2 \cup S_{\ell}$ will be a vertex cover of $G$ of size $|S|-\ell+1$ (where $S_i$ is a vertex cover of $C_i$ of size $\mathsf{mvc}(C_i)$). This is not possible as $|S^\prime|<|S|$ and $S$ is a minimum sized vertex cover of $G$. Thus there exists some $i\in [1,\ell]$ such that $|S\cap V(C_i)|=\mathsf{mvc}(C_i)$ which implies that $\{x\}$ is a weakly bad set. Hence $S$ is not a weakly good vertex cover.
\end{proof}
Thus with \Cref{wgoodVC} and \Cref{cut} we get the result in \cite{BCFPRW2021} that states that if $\mathsf{evc}(G)=\mathsf{mvc}(G)$ then each vertex must belong to a minimum sized vertex cover which contains all the cut vertices. 

We next define a generalization of a bad set, which is used to define the notion of a strongly good vertex cover. We use this notion to give another necessary condition for a graph $G$ to be Spartan. In order to define this, we define the notion of compatible vertex sets (or vertex covers) in a graph. We also give analogous definitions of compatible configurations and strongly good configuration. This gives us a new lower bound for $\mathsf{evc}(G)$. 

\begin{definition}
    Two vertex sets (vertex covers) $S_1$ and $S_2$ of a graph $G$ are said to be \emph{compatible} if there exist $|S_1\cap S_2|-$many vertex disjoint paths between $S_1\setminus S_2$ and $S_2\setminus S_1$.
\end{definition}

\begin{definition}
    Two configurations $\mathcal{C}_1$ and $\mathcal{C}_2$ are said to be \emph{compatible} if there exist $|\mathcal{C}_1\cap \mathcal{C}_2|-$many vertex disjoint paths (counting multiplicity) between $\mathcal{C}_1\setminus \mathcal{C}_2$ and $\mathcal{C}_2\setminus \mathcal{C}_1$. Here we will view each configuration as a multi-set over the set of vertices $V(G)$. 
\end{definition}
Also when we say $\mathcal{C}_i\setminus \mathcal{C}_j$ we will also account for multiplicity of each vertex. That is, if $v$ occurs $5$ times in $\mathcal{C}_1$ and $3$ times in $\mathcal{C}_2$, $v$ will occur $2$ times in $\mathcal{C}_1\setminus \mathcal{C}_2$. Also when we say vertex disjoint paths between (counting multiplicity) between $\mathcal{C}_1\setminus \mathcal{C}_2$ and $\mathcal{C}_2\setminus \mathcal{C}_1$, we mean that if a vertex $v$ occurs multiple times in $\mathcal{C}_1\setminus \mathcal{C}_2$, each occurrence of $v$ will be an endpoint of exactly one path from $\mathcal{C}_1\setminus \mathcal{C}_2$. Similarly if a vertex $v$ occurs multiple times in $\mathcal{C}_2\setminus \mathcal{C}_1$, each occurrence of $v$ will be an endpoint of exactly one path to $\mathcal{C}_2\setminus \mathcal{C}_1$. If we consider the union of all the multisets formed by the intermediate vertices of each path, the number of times a vertex $v$ occurs in the union should be less than or equal to the number of times $v$ occurs in $\mathcal{C}_1\cap\mathcal{C}_2$.

It is clear that two configurations $\mathcal{C}_1$ and $\mathcal{C}_2$ are compatible if and only if there is a possible movement of guards from $\mathcal{C}_1$ to $\mathcal{C}_2$ where each guard moves at most one step. The guards can rearrange themselves by moving one step from $\mathcal{C}_1\setminus \mathcal{C}_2$ to $\mathcal{C}_2\setminus \mathcal{C}_1$ if and only if there exist $|\mathcal{C}_1\cap \mathcal{C}_2|-$many vertex disjoint paths (counting multiplicity) between $\mathcal{C}_1\setminus \mathcal{C}_2$ and $\mathcal{C}_2\setminus \mathcal{C}_1$.

The following is a known result shown in \cite{BCFPRW2021}, but we will prove it for the sake of completeness in the appendix.
\begin{lemma}\label{lem:PM}
    Two minimum sized vertex covers of a graph $G$ are always compatible.
\end{lemma}
\begin{proof}
Let $S_1$ and $S_2$ be two minimum vertex covers of a graph $G$. Let $T_1=S_1\setminus S_2$ and $T_2=S_2\setminus S_1$. Clearly, $|T_1|=|T_2|$. Let $S_3=V(G)\setminus (S_1\cup S_2)$, i.e., $S_3$ be the intersection of the independent sets corresponding to $S_1$ and $S_2$. That is, both $T_1\cup S_3$ and $T_2\cup S_3$ are maximum independent sets of $G$. We will show that there exists a perfect matching (set of disjoint paths of length $1$) between $T_1$ and $T_2$. Suppose not, then there exits a set $X\subset T_1$ such that $|N(X)\cap T_2|<|X|$. Now consider the set $X\cup (T_2\setminus N(X))\cup S_3$. This will form a larger independent set than a maximum independent set $T_2\cup S_3$ which is a contradiction. Therefore, there exits a perfect matching between $T_1$ and $T_2$ and hence $S_1$ and $S_2$ are compatible.  
\end{proof}
    
\begin{definition}
For a vertex cover $S$, a set $T\subset V(G)\setminus S$ is \emph{strongly bad} if there is some connected component $C_i$ of $V(G)\setminus T$ and some $v\in N(T)\cap V(C_i)$ such that no vertex cover of $C_i$ of size $S\cap V(C_i)-1$ is compatible with $S\cap V(C_i)\setminus \{v\}$. 
\end{definition}
\begin{definition}
    A vertex cover $S$ of a graph $G$ is said to be \emph{strongly good} if $V(G)\setminus S$ does not have a strongly bad subset.
\end{definition}
Now we analogously define a strongly good configuration.
\begin{definition}
For a configuration $\mathcal{C}$ of guards on a vertex cover $S$, a set $T\subset V(G)\setminus S$ is \emph{strongly bad} if there is some connected component $C_i$ of $V(G)\setminus T$ and some $v\in N(T)\cap V(C_i)$ such that no configuration on a vertex cover of $C_i$ of size $\mathcal{C}\cap V(C_i)-1$ is compatible with $S\cap V(C_i)\setminus \{v\}$. 
\end{definition}
\begin{definition}
    A configuration $\mathcal{C}$ on a vertex cover $S$ of a graph $G$ is said to be \emph{strongly good} if $V(G)\setminus S$ does not have a strongly bad subset.
\end{definition}
Now we make an observation about weakly good and strongly good configurations (vertex covers).
\begin{lemma}\label{lem:sgwg}
    A strongly good configuration (vertex cover) is a weakly good configuration (vertex cover).
\end{lemma}
\begin{proof}
We will prove the contrapositive of this statement, i.e., a configuration on a vertex cover which is not a weakly good configuration is not a strongly good configuration. Let $\mathcal{C}$ be a configuration on a vertex cover $S$ such that $\mathcal{C}$ is not weakly good. Therefore, there exists a non-empty $T\subseteq I=V(G)\setminus S$ such that if $C_1,C_2,\ldots C_\ell$ are the connected components of $G[V(G)\setminus T]$, there must exist a $C_i$ (where $i\in [1,\ell]$) such that the number of guards in $C_i$ (in $\mathcal{C}$) must be equal to $\mathsf{mvc}(C_i)$. Since $T$ is non-empty and the graph $G$ is connected, there exists $v\in V(C_i)$ such that $v$ is adjacent to $T$ and since $S$ is a vertex cover $v$ is in $S$ (and has exactly one guard in $\mathcal{C}$ as the number of guards in $C_i$ in $\mathcal{C}$ is equal to $\mathsf{mvc}(C_i)$). The configuration $\mathcal{C}\cap C_i\setminus \{v\}$ is not compatible with any vertex cover of $C_i$ as the number of guards in $\mathcal{C}\cap C_i\setminus \{v\}$ is less than $\mathsf{mvc}(C_i)$. Thus the set $T$ is a strongly bad set and $\mathcal{C}$ is not a strongly good configuration. 

The same proof works for vertex covers if we consider a vertex cover as a configuration with one guard per vertex.
\end{proof}

\begin{lemma}\label{sgoodVC}
A graph $G$ has $\mathsf{evc}(G)=k$ only if for each $v\in V(G)$, there exists a $k-$sized strongly good configuration $\mathcal{C}_v$ on a vertex cover $S_v$ such that $v\in S_v$. In particular, a graph $G$ is Spartan if for each $v\in V(G)$, there exists a minimum sized and strongly good vertex cover $S_v$ such that $v\in S_v$. 
\end{lemma}
Suppose that there exists a graph $G$ with $\mathsf{evc}(G)=k$ and a vertex $v$ such that there exists no $k-$sized strongly good configuration on a vertex cover containing $v$. Without loss of generality we can assume that the initial configuration has at least one guard on $v$ because if not, the attacker can force a guard to come to $v$ by attacking an edge adjacent to $v$. If the vertices occupied by the guards in the resulting configuration do not form a vertex cover, then there exists some edge which is vulnerable and hence the attacker wins. Therefore, the resulting configuration must have guards on a vertex cover. Since there is no $k-$sized strongly good configuration $\mathcal{C}_v$ on a vertex cover $S_v$ such that $v\in S_v$, for the configuration $\mathcal{C}$ formed by the guards on the vertex cover $S$, there exists a strongly bad subset $T\subset V(G)\setminus S$. Hence there is some connected component $C_i$ of $V(G)\setminus T$ and some $u\in N(T)\cap V(C_i)$ such that no configuration on a vertex cover of $C_i$ of size $\mathcal{C}\cap V(C_i)-1$ is compatible with $S\cap V(C_i)\setminus \{u\}$. Suppose the attacker attacks an edge joining $u$ to a vertex in $T$, the guard on $u$ is forced to move to $T$. No guard from a vertex out side $C_i$ can come to $C_i$ and the guards in $C_i$ cannot rearrange themselves to form a vertex cover of $C_i$. Therefore some edge will be vulnerable no matter how the guards rearrange themselves which contradicts $\mathsf{evc}(G)=k$.

Therefore, a graph $G$ has $\mathsf{evc}(G)=k$ only if for each $v\in V(G)$, there exists a $k-$sized strongly good configuration $\mathcal{C}_v$ on a vertex cover $S_v$ where $v\in S_v$. 

\section{K\"onig Graphs}\label{sec:konig}

A graph $G$ is said to be a \emph{K\"onig} graph if the size of the smallest vertex cover of $G$ is equal to the size of the largest matching of $G$, i.e. $\mathsf{mvc}(G)=mm(G)$. This class is a natural and strict{\footnote{There are K\"onig graphs which are non-bipartite as well. For example consider $G$ with $V(G)=\{a_1,b_1,a_2,b_2\}$ and $E(G)=\{a_1a_2,a_1b_1,a_2b_2,a_1b_2,a_2b_1\}$. This graph has a perfect matching $\{a_1b_1,a_2b_2\}$ and thus $mm(G)=2$ and thus $\mathsf{mvc}(G)\geq 2$. But $\{a_1,a_2\}$ forms a vertex cover of $G$ of size $2$ thus $\mathsf{mvc}(G)=mm(G)$ and $G$ is not bipartite because it contains a triangle.}}  generalization of bipartite graphs. In this section, we derive the necessary and sufficient condition for a K\"onig Graph to be Spartan. Before that we will show a corollary of \Cref{wgoodVC} which gives a sense of how a Spartan graph will look like.

\begin{corollary}\label{nbdI}
   If a graph $G$ with more than one vertex is Spartan and $I$ is an independent set of $G$ which is not maximal, then $|N(I)|>|I|$.  
\end{corollary}
\begin{proof}
    Let $G$ be a graph with more than one vertex which is Spartan and let $I$ be an independent set of $G$ which is not maximal such that $|N(I)|=|I|$. Denote $N(I)$ by $C$. Consider an inclusion wise minimal such set $I$. The graph $H$ with vertex set $I\cup C$ and edge set $E(G[I\cup C])\setminus E(C)$ must be bipartite (because $I$ is an independent set) and elementary (using the inclusion wise minimality of $I$ and \Cref{degone}). Therefore, using an alternate definition of elementary bipartite graph from \cite{LLPM-matchingbook}, $I$ and $C$ are the only two minimum sized vertex covers of $H$. Now, any minimum vertex cover $S$ if $G$ cannot contain more than $|I|$ many vertices from $G[I\cup C]$. Suppose not, $S\setminus (I\cup C)\cup C$ is a strictly smaller vertex cover of $G$ (this is a vertex cover as there are no edges from $I$ to $V(G)\setminus C$ and hence $S\setminus C$). 

Now, since $G$ is Spartan, by \Cref{wgoodVC}, every vertex in $I$ belongs to a minimum sized weakly good vertex cover. Consider a minimum sized vertex cover $T$ which contains some $v\in I$. Clearly, $T$ can contain only $|I|$ vertices from $C\cup I$. Also as $T$ is also a vertex cover of $H$ of size $|I|$ and $H$ has only two minimum vertex covers, $I\subset T$ and $C\cap T=\emptyset$. The graph $G[V(G)\setminus(I \cup C)])$ is non-empty as $I$ is not maximal. If $T\cap (V(G)\setminus(I \cup C))$ is more than $mvc(G[V(G)\setminus(I \cup C)])$ then a smaller vertex cover of $G$ than $T$ can be obtained from the union of $C$ and a minimum vertex cover of  $(V(G)\setminus(I \cup C))$. Thus $T\cap (V(G)\setminus(I \cup C))$ is equal to the $mvc(G[V(G)\setminus(I \cup C)])$ which makes $C$ a weakly bad subset for $T$ which is a contradiction.
\end{proof}
The main result of this section is the following:
\konigspartan*
\begin{proof}
As described in \Cref{connected}, it is sufficient to look at connected graphs i.e. it will be sufficient to prove that a  connected K\"onig graph $G$ is Spartan if and only if it is bipartite and elementary. The reverse direction has already been proved in \cite{AFI2015}. Thus we need to only show that if a connected K\''onig graph $G$ is Spartan, then it is bipartite and elementary. 

Since $G$ is Spartan, let $S$ be a minimum sized vertex cover of $G$ which is an eternal vertex cover configuration and let $I=V(G)\setminus S$ be the corresponding independent set. By \Cref{matchI}, there exists a matching $M$ from $I$ to $S$ which saturates $I$. This means that $|S|\geq |I|$. Since $S$ is a minimum sized vertex cover, each edge of a maximum matching of $G$ must have at least one endpoint in $S$. Suppose some edge of a maximum matching has two endpoints in $S$, then $|S|>mm(G)$, which contradicts the assumption that $G$ is a K\"onig graph. Therefore every edge of a maximum matching of $G$ has exactly one endpoint in $S$. Thus, $|S|\leq |V(G)\setminus S|$ i.e., $|S|\leq |I|$. Hence we obtain $|S|=|I|=\nicefrac{V(G)}{2}$ and $M$ is a perfect matching of $G$.
   
Now consider the bipartite graph $H$ such that $V(H)=V(G)$ and $E(H)=E(G)\setminus E(G[S])$ i.e. $H$ has the same vertex set as that of $G$ and edge set of $H$ consists of edges of $G$ going across from $S$ to $I$. Now $V(H)=S\cup I$ such that $S$ and $I$ are independent sets in $H$ and $M$ is a perfect matching of $H$. Suppose there exists $T\subsetneq I$ such that $|N(T)|=|T|$ in $H$. Since $I$ is an independent set in $G$ as well, we have $|N(T)|=|T|$ in $G$. Since $T\subsetneq I$, $T$ is an independent set of $G$ which is not maximal. Thus by \Cref{nbdI}, we get that $G$ is not Spartan which is a contradiction. Therefore $H$ must be elementary by \Cref{degone}. 

Thus we know that $H$ can have only two minimum sized vertex covers $S$ and $I$ (using an alternate definition of elementary bipartite graph from \cite{LLPM-matchingbook}). Since $E(H)\subseteq E(G)$, any vertex cover of $G$ must be a vertex cover of $H$. Thus $G$ cannot have any minimum sized vertex cover other than $S$ and $I$. If there is an edge of $G$ with both its endpoints in $S$, then $I$ is not a vertex cover of $G$ and thus $S$ is the only minimum sized vertex cover of $G$. Let $v\in I$, there does not exist a minimum sized vertex cover of $G$ which contains $v$. Thus using \Cref{vinVC}, we get that $G$ is not Spartan which is a contradiction. Thus there cannot be any edge with both the endpoints in $S$. Thus $E(G)=E(H)$ and we already know that $V(G)=V(H)$ and $H$ is bipartite and elementary. Therefore, $G$ is bipartite and elementary.  
\end{proof}

\section{General Graphs}
\label{sec:general}

Before stating the main result of this section, we will need the definition of an auxiliary graph defined for a pair of vertex covers of a graph $G$. Also, in this section, as we are dealing with Spartan-ness, we are dealing with configurations of guards of size $\mathsf{mvc}(G)$. If there is more than one guard on some vertex of a configuration, the set of vertices occupied by guards does not form a vertex cover and this configuration can be attacked by the attacker immediately. Hence we look at only those configurations with one guard per vertex.

\begin{definition}
Let $S$ and $T$ be two vertex covers of a graph $G$. Let $H_1, \ldots, H_k$ be the connected components of $G[S \cap T]$. For $i, 1 \leqslant i \leqslant k$, we say that two vertices $u \in S\setminus T$ and $v \in T\setminus S$ are \emph{pseudo-adjacent via $i$} if both $u$ and $v$ are adjacent to some vertex in $V(H_i)$.

We define two subsets of $D \times D$:

$$E_1 := \{(u,v) \mid u \in S \setminus T, v \in T \setminus S, uv \in E(G)\},$$

and 

$$E_2 := \{(u,v) \mid u \in S \setminus T, v \in T \setminus S, uv \text{ are pseudo-adjacent via }i \text{ for some } i\}.$$

We use $\mathfrak{h}_G(S,T)$ to denote the graph $(D,E_1 \cup E_2)$. 
\end{definition}

\begin{figure}
    \centering
    \begin{tikzpicture}[scale=0.77]
    \tikzset{decoration={snake,amplitude=.4mm,segment length=2mm, post length=0mm,pre length=0mm}}
       \draw[fill=DodgerBlue!70] (0,0) circle (0.3 cm); 
       \draw[fill=DodgerBlue!70] (0,1) circle (0.3 cm); 
       \draw[fill=DodgerBlue!70] (0,2) circle (0.3 cm); 
       \draw[fill=DodgerBlue!70] (0,3) circle (0.3 cm); 

      \draw[fill=CarnationPink!70] (2,0) circle (0.3 cm); 
      \draw[fill=CarnationPink!70] (2,1) circle (0.3 cm); 
      \draw[fill=CarnationPink!70] (2,2) circle (0.3 cm); 
      \draw[fill=CarnationPink!70] (2,3) circle (0.3 cm); 

      \node at (0,3.8){$S\setminus T$};
      \node at (2,3.8){$T\setminus S$};
      
     \draw [RoyalBlue,thick,rounded corners] (-0.5,-0.5) rectangle (0.5,3.5);
     \draw [WildStrawberry,thick,rounded corners] (1.5,-0.5) rectangle (2.5,3.5);

      \draw[fill=OrangeRed!70] (-1,-2) circle (0.7 cm); 
      \node at (-1,-2){$C_1$};
      \draw[fill=SeaGreen] (1,-3) circle (0.7 cm); 
      \node at (1,-3){$C_2$};
      \draw[fill=Purple!70] (3,-2) circle (0.7 cm); 
      \node at (3,-2){$C_3$};

       \draw[fill=DodgerBlue!70] (6,0) circle (0.3 cm); 
       \draw[fill=DodgerBlue!70] (6,1) circle (0.3 cm); 
       \draw[fill=DodgerBlue!70] (6,2) circle (0.3 cm); 
       \draw[fill=DodgerBlue!70] (6,3) circle (0.3 cm); 

      \draw[fill=CarnationPink!70] (8,0) circle (0.3 cm); 
      \draw[fill=CarnationPink!70] (8,1) circle (0.3 cm); 
      \draw[fill=CarnationPink!70] (8,2) circle (0.3 cm); 
      \draw[fill=CarnationPink!70] (8,3) circle (0.3 cm); 

      \draw [RoyalBlue,thick,rounded corners] (5.5,-0.5) rectangle (6.5,3.5);
      \draw [WildStrawberry,thick,rounded corners] (7.5,-0.5) rectangle (8.5,3.5);

      \node at (6,3.8){$S\setminus T$};
      \node at (8,3.8){$T\setminus S$};

      \draw[thick](0.3,3)--(1.7,3);
      \draw[thick](0.3,2)--(1.7,2);
      \draw[thick](0.3,2)--(1.7,3);
      \draw[thick](0.3,1)--(1.7,2);
      \draw[thick](0.3,1)--(1.7,0);
      \draw[thick,dashed] (-0.3,1)--(-1.5,-1.5);
      \draw[thick,dashed] (0,-0.3)--(-1,-1.3);
      \draw[thick,dashed] (2,0.7)--(-0.6,-1.4);
      \draw[thick,dashed] (0.3,1)--(0.7,-2.3);
      \draw[thick,dashed] (1.7,1)--(1.3,-2.3);
      \draw[thick,dashed] (2,-0.3)--(1.5,-2.5);
      \draw[thick,dashed] (0.3,3)--(2.7,-1.3);
      \draw[thick,dashed] (2.3,3)--(3.2,-1.3);
      \draw[thick](6.3,3)--(7.7,3);
      \draw[thick](6.3,2)--(7.7,2);
      \draw[thick](6.3,2)--(7.7,3);
      \draw[thick](6.3,1)--(7.7,2);
      \draw[thick](6.3,1)--(7.7,0);    

      \draw[thick,decorate, OrangeRed](6.3,1)--(7.7,1);
      \draw[thick,decorate, OrangeRed](6.3,0)--(7.7,1);
      \draw[thick,decorate, SeaGreen](6,1.3)--(8,1.3);
      \draw[thick,decorate, SeaGreen](6,0.7)--(7.7,0);
      \draw[thick,decorate, Purple](6,3.3)--(8,3.3);
    \end{tikzpicture}
    \caption{Depicting the construction of $\mathfrak{h}_G(S,T)$.}
    \label{fig:1}
\end{figure}

Note that $\mathfrak{h}_G(S,T)$ may have multiedges. In the context of this graph, we refer to the edges in $E_1$ as \emph{real} edges and the edges in $E_2$ as \emph{helper} edges. 

Note that the helper edges can be naturally partitioned into $k$ sets $E_{2}^{(1)}, \ldots, E_{2}^{(k)}$, where $E_{2}^{(i)}$ consists of the helper edges that are pseudo-adjacent via $i$. For convenience, we refer to the edges in $E_{2}^{(i)}$ as being edges of color $i$. 

The following can be easily observed:
\begin{lemma}\label{lem:bip}
    The graph $h_G(S,T)$ is bipartite for a graph $G$ with at least one edge.
\end{lemma}
\begin{proof}
    Since $S\setminus T$ and $T\setminus S$ are subsets of independent sets $V(G)\setminus T$ and $T\setminus S$ respectively, the graph $h_G(S,T)$ is bipartite.
\end{proof}

\begin{restatable}{theorem}{spartan}
    A graph $G$ is Spartan if and only if there exists a non-empty family $\mathcal{F}$ of minimum-sized vertex covers such that the following conditions hold:

    For every $S \in \mathcal{F}$ and for every edge $uv \in G$ for which $u \in S$ and $v \notin S$, there exists $T \in \mathcal{F}$ such that and $v \in T$ such that:
        \begin{enumerate}[(1)]
        \item either $u \notin T$ and there is perfect matching in $\mathfrak{h}_G(S,T)$ which contains the edge $uv$, 
        \item or $u \in T$ and there is a perfect matching in $\mathfrak{h}_G(S,T)$ where the matched partner of $v$, say $w$, has a neighbor --- in $G$ --- among the vertices of $X$, where $X$ is the connected component of $G[S \cap T]$ that contains $u$. 
        \end{enumerate}  
\end{restatable}

\begin{proof} First, assume that $G$ is Spartan. Then there is a strategy for indefinite defense of $G$ with $\mathsf{mvc}(G)$ guards. Let $\mathcal{F}$ be the set of all vertex covers that are used in the strategy: note that they are all of minimum size by definition. Since $G$ is Spartan, $\mathcal{F}$ is indeed non-empty.
    


Consider a vertex cover $S\in \mathcal{F}$ and an edge $uv$ such that $u \in S$ and $v \notin S$.

As $S$ occurs in the strategy of the defender, there is a way to defend an attack on an edge when guards occupy the vertices of $S$ (one guard on each vertex). We attack the edge $uv$ here and observe the promised defense: let us say that the guards are positioned on the vertex cover $T$ after the defense is executed. Clearly, $T \in \mathcal{F}$, by definition. Depending on whether $u \in T$ or not, we can conclude that either (1) holds or (2) does (respectively), by tracing the movement of the guards from vertices in $S \setminus T$ to vertices in $T \setminus S$. 

Suppose $u\notin T$, we show that (1) holds. Let $S'=S\setminus T$, i.e, the set of vertices which had a guard before the attack and do not have a guard after the defense is completed. Let $T'=T\setminus S$, i.e, the set of vertices which had a guard after the defense is completed and do not have a guard before the attack. Let $|S'|=|T'|=k$. Then there must be $k$ vertex disjoint paths $\mathcal{P}$ from $S'$ to $T'$ (with the starting vertex of each path from $S'$, end vertex in $T'$ and each intermediate vertex from $S\cap T$) and one of these paths must be the edge $uv$. We show how the collection $\mathcal{P}$ is obtained. Each path is obtained by tracing the movement of each individual guard. The guard on $u$ is forced to move to $v$ by the attacker and hence the edge $uv$ is traced by a guard. Thus $uv\in \mathcal{P}$. Similarly look at the movement of a guard on a vertex say $u_1\neq u$ of $S'$. This guard moves to some vertex $u_2$. If $u_2\in T$, $u_1u_2\in \mathcal{P}$. Otherwise $u_2\in S\cap T$ as the guard can only move for one step after each attack and hence $u_2$ must have a guard both before and after the attack. The guard which previously was on $u_2$ moves to some $u_3$ (as both $S$ and $T$ are vertex covers hence there must be only one guard per vertex after the reconfiguration has been done). Now again if $u_3\in T$, $u_1u_2u_3\in \mathcal{P}$. Otherwise $u_3\in S\cap T$ and the guard which was previously on $u_3$ must move to some $u_4$. This process will only stop when a guard moves to a vertex which did not already have a guard, i.e., a vertex in $T'$. We obtain $k$ paths by tracing the movement of each of the $k$ guards in $S'$ (including the guard moving from $u$ to $v$). It is clear that a vertex $v$ in $S'$ or $T'$ cannot belong to two paths because this will indicate that two guards started from the vertex $v$ (if $v\in S'$) or two guards ended up at $v$ (if $v\in T'$). A vertex $v\in S\cap T$ cannot belong to two paths in $\mathcal{P}$ because this will indicate that this vertex has at least two guards after the reconfiguration is done which is not possible. Also, each path in $\mathcal{P}$ which contains more than one edge must have all the intermediate vertices from the same connected component of $S\cap T$. This is because the intermediate vertices contain a guard both before and after the reconfiguration and hence lie in $S\cap T$. Each guard can move only to a neighboring vertex, the intermediate vertices of a path in $\mathcal{P}$ also form a path in $S\cap T$. Thus a path $u_1u_2\ldots u_t\in \mathcal{P}$ where $t>2$ implies $u_1\in S'$, $u_t\in T'$ and $u_2,u_3,\ldots,u_{t-1}\in V(H_i)$ for some $i$. Therefore, in the graph $\mathfrak{h}_G(S,T)$, there exists a helper edge of color $i$ between $u_1$ and $u_t$. Also an edge in $\mathcal{P}$ corresponds to a real edge in  $\mathfrak{h}_G(S,T)$. Thus each path in $\mathcal{P}$ gives an edge in $\mathfrak{h}_G(S,T)$. The edges in $\mathfrak{h}_G(S,T)$ corresponding to the paths in $\mathcal{P}$ form a perfect matching because each vertex in $S'$ is adjacent to exactly one edge corresponding to a path in $\mathcal{P}$ because there is only one guard on each vertex of $S'$ before the attack and each vertex in $S'$ has no guard after the attack. Thus $\mathfrak{h}_G(S,T)$ contains a perfect matching containing the edge $uv$ and condition 2(a) holds.

Now suppose $u\in T$, we show that (2) holds. Let $S'$ and $T'$ be the same as above and similarly $k=|S'|=|T'|$. By the same reasoning as the previous case, we have $k$ vertex disjoint paths in $G$ which correspond to a perfect matching $M$ in $\mathfrak{h}_G(S,T)$. Now as $u\in S$ and $u\in T$, a guard is present on $u$ before the attack and after the reconfiguration in complete. Since $T$ is obtained from $S$ by applying the winning strategy for a defense after the attacker attacks the edge $uv$, the guard on $u$ must move to $v$ while reconfiguring from $S$ to $T$. Since a guard is present on $u$ after the attack, there must exist a path $u_1u_2\ldots u_t$ in $\mathcal{P}$ where $t>2$ such that $u_1\in S'$, $u_{t-1}=u$ and $u_t=v$. Also, $u_2,u_3,\ldots,u_{t-1}$ belong to the same connected component of $H_i$. This path corresponds to the movement: A guard on a vertex $u_1$ of $S'$ moves to a vertex $u_2$ of $S\cap T$. The guard on $u_2$ moves to $u_3$ of $S\cap T$ and so on till the guard on $u_{t-2}$ moves to $u$ and the guard on $u$ moves to $v$. Clearly the vertices $u_2,u_3,\ldots,u_{t-1}$ must belong to the same connected component of $S\cap T$ because a guard can only move to a neighboring vertex. The edge $u_1v$ is a helper edge in $ \mathfrak{h}_G(S,T)$ and lies in $M$. The vertex $u_1$ has a neighbour $u_2$ in the connected component of $G[S\cap T]$ that contains $u$ (because there is a path $u_2u_3\ldots u_{t-1}=u$ in $G[S\cap T]$. Hence (2) holds.

Thus we have shown that when $G$ is Spartan, the family $\mathcal{F}$ of all vertex covers used in a strategy of the defender; satisfies both Conditions (1) and (2).

Now we show the converse, i.e., if a graph $G$ has a family of minimum sized vertex covers $\mathcal{F}$ which satisfy both (1) and (2), then $G$ is Spartan. For this we show that if the guards occupy a vertex cover $S\in \mathcal{F}$ and an arbitrary edge $uv$ is attacked, the guards can reconfigure themselves such that at least one guard moves across the attacked edge $uv$ and the final positions of the guards form a vertex cover $T\in \mathcal{F}$.

Since $S$ is a vertex cover, there cannot be any edge with both the endpoints in $S$. If an edge $uv$ such that $u,v\in S$ is attacked, the guard on $u$ moves to $v$ and the guard on $v$ moves to $u$. The attack is defended and the configuration $S$ is restored. Therefore we need to consider the attack on an edge $uv$ such that $u\in S$ and $v\notin S$.

By the given condition, there exists a vertex cover $T \in \mathcal{F}$ such that and $v \in T$ such that:
        \begin{enumerate}[(1)]
        \item either $u \notin T$ and there is perfect matching in $\mathfrak{h}_G(S,T)$ which contains the edge $uv$, 
        \item or $u \in T$ and there is a perfect matching in $\mathfrak{h}_G(S,T)$ where the matched partner of $v$, say $w$, has a neighbor --- in $G$ --- among the vertices of $X$, where $X$ is the connected component of $G[S \cap T]$ that contains $u$. 
        \end{enumerate}

Suppose (1) holds. We will show that it is possible to reconfigure the guards from $S$ to $T$ such that one guard moves across $uv$. Let $S'=S\setminus T$, $T'=T\setminus S$ and $|S'|=|T'|=k$. It is enough to show that there is a collection $\mathcal{P}$ containing $uv$ of $k$ vertex disjoint paths from $S'$ to $T'$ with all the intermediate vertices in $S\cap T$. A path $u_1u_2\ldots u_t$ in $G$ where $u_1\in S'$, $u_t\in T'$ and $u_2,u_3,\ldots,u_{t-1}\in S\cap T$ will represent the movement of a guard from $u_1$ to $u_2$, the movement of the guard previously on $u_2$ to $u_3$, and so on up to the movement of the guard previously on $u_{t-1}$ to $u_t$. 

If there exists a perfect matching $M$ in $\mathfrak{h}_G(S,T)$ such that it consists of possibly some real edges and at most one helper edge of each color, then there exists a collection $\mathcal{P}$ of $k$ vertex disjoint paths from $S'$ to $T'$ with intermediate vertices in each path from $S\cap T$. We show that a path in $\mathcal{P}$ can be obtained from each edge of $M$. A real edge in $\mathfrak{h}_G(S,T)$ is also an edge in $G$. Thus for each real edge $e$ in $M$, add $e$ to $\mathcal{P}$. A helper edge $uv$ of color $i$ implies the existence of a path  $u_1=uu_2u_3\ldots u_t=v$ in $G$ where $t>1$ and $u_2,u_3,\ldots,u_{t-1}\in V(H_i)$ where $H_i$ is a connected component of $S\cap T$. For each helper edge $e=wz$ in $M$ of color $i$, add a path $u_1=wu_2u_3\ldots u_t=z$ in $G$ where $t>1$ and $u_2,u_3,\ldots,u_{t-1}\in V(H_i)$ to $\mathcal{P}$. Clearly, there are $k$ paths in $\mathcal{P}$ (one path obtained from each edge of $M$). We show that the paths in $\mathcal{P}$ are vertex disjoint. Since $M$ is a matching, the endpoints of each edge in $M$ are distinct. Hence the endpoints of all the $k$ paths are distinct. Since all the intermediate vertices of each path are distinct and each path of length at least $2$ is obtained from a helper edge of different color (as there are no two helper edges of the same color in $M$), each path obtained from a helper edge in $M$ has intermediate vertices from distinct components of $G[S\cap T]$. Thus all the intermediate vertices of each path in $\mathcal{P}$ are disjoint.

Now we show that there exists a perfect matching in $\mathfrak{h}_G[S,T]$ which contains possibly some real edges and at most one helper edge of each color. By \Cref{lem:PM}, there exists a perfect matching from $S'$ to $T'$ in $G$. This means that there exits a perfect matching $M_p$ in $\mathfrak{h}_G(S,T)$ which consists of only real edges. Let $M$ be the perfect matching in $\mathfrak{h}_G(S,T)$ such that $M$ contains $uv$ (exists by condition 2$(a)$) and the size of $M\cap M_{p}$ is as large as possible. Now consider the graph $H$ in $\mathfrak{h}_G(S,T)$ such that $V(H)=V(\mathfrak{h}_G(S,T))$ and $E(H)=M\cup M_{p}$. Since $M$ and $M_p$ are both perfect matchings, $H$ will be a union of edges and even cycles.

\begin{claim}
   There can be at most one cycle in $H$.
\end{claim}

\begin{proof}
    If there are two distinct cycles $C_1$ and $C_2$ in $M$, at most one of them can contain the edge $uv$. Therefore, without loss of generality we can assume that $uv\notin E(C_1)$. Suppose $C_1$ is a cycle and $C_1=u_1v_1u_2v_2\ldots u_tv_tu_1$ where $M_1:=\{u_1v_1,u_2v_2,\ldots,u_tv_t\}\subset M$ and $M_2:=\{v_1u_2,v_2u_3,\ldots,v_{t-1}u_t,v_tu_1\}\subset M_p$ with $t>1$. The matching $M':=M\setminus M_1\cup M_2$ will have strictly more intersection with $M_p$ than $M$ and will also contain the edge $uv$ which is a contradiction.
\end{proof}
If $H$ has no cycle then $M=M_p$ which means that all the edges of $M$ are real and we are done. Now we consider the case where $M$ has exactly one even cycle $C$. If $uv\notin E(C)$, then we get a contradiction by the same argument as the claim above. Therefore let $C=u_1(=u)v_1(=v)u_2v_2\ldots u_tv_tu_1$. Here $M_1:=\{u_1v_1,u_2v_2,\ldots,u_tv_t\}\subset M$ and $M_2:=\{v_1u_2,v_2u_3,\ldots,v_{t-1}u_t,v_tu_1\}\subset M_p$. Also, $\{u_1, u_2,\ldots,u_t\}\subset S'$ and $\{v_1, v_2,\ldots,v_t\}\subset T'$. This is because $u\in S'$ and $\mathfrak{h}_G(S,T)$ is bipartite by \Cref{lem:bip}.  All the edges in $M$ which are not in $M_1$ are also the edges of $M_p$ and hence are real. Next we show that $M$ cannot contain two (or more) helper edges of the same color.
\begin{claim}
    $M$ can contain at most one helper edge of each color.
\end{claim}
\begin{proof}
    Suppose $M$ contains two helper edges of the same color (say $i$). These two edges must lie in $M_1$ because the edges of $M$ outside $M_1$ are real. Let these two edges be $u_pv_p$ and $u_qv_q$ where $1<p<q$. (As $uv$ is a real edge, we have $p\neq 1$.) Therefore there must exist edge $u_pv_q$ of color $i$. This is because $u_p$ and $v_p$ are both adjacent to some vertex in $V(H_i)$ and similarly $u_q$ and $v_q$ are both adjacent to some vertex in $V(H_i)$. Hence $u_p$ and $v_q$ are adjacent to some vertex in $V(H_i)$. Also, $M_3:=\{v_pu_{p+1},v_{p+1}u_{p+2},\ldots,v_{q-1}u_q\}\subset M_p$ hence the edges in $M_3$ are real edges of $\mathfrak{h}_G(S,T)$. Let $M_4=\{u_pv_p,u_{p+1}v_{p+1},\ldots,u_qv_q\}$ and $M_5=(M_1\setminus M_4)\cup M_3\cup \{u_pv_q\}$. Consider the perfect matching $M'=(M\setminus M_1)\cup M_5$. It can be checked that $M'$ is a perfect matching which contains $uv$ and has greater intersection with $M_p$ than $M$ which is a contradiction.
\end{proof}

\begin{figure}
    \centering
    \begin{tikzpicture}[scale=0.77]
    \tikzset{decoration={snake,amplitude=.4mm,segment length=2mm, post length=0mm,pre length=0mm}}
    \draw[thick] (6.25, 3.25)--(6.75, 3.75);
\draw[thick] (6.8, 5.65)--(6.25, 6.25);
\draw[thick] (4.3, 6.25)--(3.75, 5.65);
\draw[thick] (3.75, 3.75)--(4.25, 3.25);
\draw[thick,decorate,Fuchsia] (4.8,3)--(5.7,3);
\draw[thick,red] (7,4.3)--(7,5.1);
\draw[thick,decorate,SeaGreen] (5.7,6.4)--(4.8,6.4);
\draw[thick,decorate,Fuchsia] (3.5,4.3)--(3.5,5.1);

\draw[thick] (9.75, 3.75)--(10.25, 3.25);
\draw[thick,red] (13,4.3)--(13,5.1);
\draw[thick,decorate,SeaGreen] (11.7,6.4)--(10.8,6.4);
\draw[thick,decorate,Fuchsia] (11.7,3.2)--(9.6,5.25);

\draw[fill=YellowGreen] (6, 3) circle (0.3cm);
\draw[fill=DodgerBlue] (7, 4) circle (0.3cm);
\draw[fill=YellowGreen] (7, 5.4) circle (0.3cm);
\draw[fill=DodgerBlue] (6, 6.4) circle (0.3cm);
\draw[fill=YellowGreen] (4.5, 6.4) circle (0.3cm);
\draw[fill=DodgerBlue] (3.5, 5.4) circle (0.3cm);
\draw[fill=YellowGreen] (3.5, 4) circle (0.3cm);
\draw[fill=DodgerBlue] (4.5, 3) circle (0.3cm);
\draw[fill=YellowGreen] (12, 3) circle (0.3cm);
\draw[fill=DodgerBlue] (13, 4) circle (0.3cm);
\draw[fill=YellowGreen] (13, 5.4) circle (0.3cm);
\draw[fill=DodgerBlue] (12, 6.4) circle (0.3cm);
\draw[fill=YellowGreen] (10.5, 6.4) circle (0.3cm);
\draw[fill=DodgerBlue] (9.5, 5.4) circle (0.3cm);
\draw[fill=YellowGreen] (9.5, 4) circle (0.3cm);
\draw[fill=DodgerBlue] (10.5, 3) circle (0.3cm);

\node at (6,3){$a_1$};
\node at (7,4){$b_1$};
\node at (7,5.4){$a_2$};
\node at (6,6.4){$b_2$};
\node at (4.5,6.4){$a_3$};
\node at (3.5,5.4){$b_3$};
\node at (3.5,4){$a_4$};
\node at (4.5,3){$b_4$};
\node at (12,3){$a_1$};
\node at (13,4){$b_1$};
\node at (13,5.4){$a_2$};
\node at (12,6.4){$b_2$};
\node at (10.5,6.4){$a_3$};
\node at (9.5,5.4){$b_3$};
\node at (9.5,4){$a_4$};
\node at (10.5,3){$b_4$};
    
    \end{tikzpicture}
    \caption{Demonstrating the exchange argument from the proof.}
    \label{fig:exchange}
\end{figure}
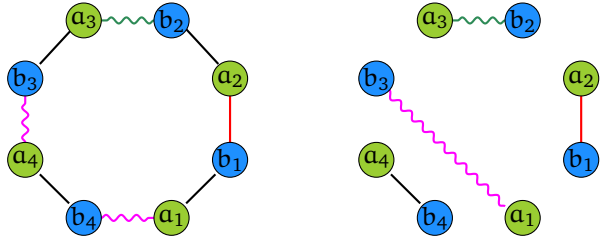

    Thus we have shown the existence of a perfect matching between $S'$ and $T'$ containing the edge $uv$ in $\mathfrak{h}_G(S,T)$ which contains at most one helper edge of each color. This implies that there is a collection of $|S'|$ many vertex disjoint paths from $S'$ to $T'$ such that one path in this collection is the edge $uv$. Therefore, the guards can reconfigure from $S$ to $T$ (such that one guard moves across $uv$).

    Now suppose (2) holds. Since $w$ has a neighbor in the same connected component $X=C_i$ (say) of $G[S\cap T]$ that contains $u$, this means that there exists a path $w=u_1u_2\ldots u_{t-1}=u u_t=v$ where $t>1$ such that $u_2,u_3,\ldots, u_{t-1}(=u) $ belong to $V(C_i)$ that is, $wv$ is a helper edge of color $i$ in $\mathfrak{h}_G(S,T)$. As we have seen above, a matching in $\mathfrak{h}_G(S,T)$ with at most one helper edge of each color can be used to show the existence of $S\setminus T$ many vertex disjoint paths from $S'=S\setminus T$ to $T'=T\setminus S$ with the intermediate vertices from $S\cap T$. This shows that it is possible to reconfigure the guards from $S$ to $T$.

    Now if we have a perfect matching in $\mathfrak{h}_G(S,T)$ containing $wv$ with at most one helper edge of each color, it can be seen that it is possible to reconfigure the guards from $S$ to $T$ such that at least one guard moves across $uv$. This is because we have already seen that such a matching represents vertex-disjoint paths in $G$ and particularly, the edge $wv$ represents the path $(w=)u_1u_2\ldots u_{t-1}(=u) u_t(=v)$. This means that it is possible to move the guard on $w$ to $u_2$, the guard previously on $u_2$ to $u_3$ and so on up to the guard previously on $u_{t-2}$ to $u$ and the guard previously on $u$ to $v$. Thus, the existence of a perfect matching in $\mathfrak{h}_G(S,T)$ containing $wv$ with at most one helper edge of each color is enough to show that the defender can defend an attack on the edge $uv$ and reconfigure the guards to a vertex cover $T$ while starting from a vertex cover $S$. It remains to show the existence of such a perfect matching in $\mathfrak{h}_G(S,T)$.

    By \Cref{lem:PM}, we know that there exists a perfect matching $M_p$ between $S'$ and $T'$ in $\mathfrak{h}_G(S,T)$ which consists of only real edges. Let $M$ be a perfect matching which contains a helper edge of color $i$ adjacent to $v$ (where $u\in V(C_i)$ for the component $C_i$ of $G[S\cap T]$) such that $M\cap M_p$ is as large as possible. It can be seen by a similar argument to the case (1) that the graph $H$ with $V(H)=S'\cup T' $ and $E(H)=M\cup M_p$ can have at most one cycle (even length). 

    If $H$ has no cycle, then $M=M_p$ which is not possible because $M$ has only real edges and $M_p$ has at least one helper edge of color $i$. Let the cycle in $H$ be $C=u_1v_1u_2v_2\ldots u_tv_tu_1$ where $M_1:=\{u_1v_1,u_2v_2,\ldots,u_tv_t\}\subset M$ and $M_2:=\{v_1u_2,v_2u_3,\ldots,v_tu_1\}\subset M_p$. Again, if $wv\notin M_1$, then $M'=M\setminus M_1 \cup M_2$ is a perfect matching containing a helper edge $wv$ (which is of color $i$ and adjacent to $v$) and $M'\cap M_p$ has greater size than $M\cap M_p$ which is a contradiction. Therefore, $wv\in M_1$. Also, if $M_1$ contains any two helper edges of the same color (other than $wv$), then we can use an exchange argument similar to the case (1) shown above to get a contradiction.

    It remains to show that $C$ cannot have an edge of color $i$ other than $wv$. Without loss of generality, let $w=u_1$ and $v=v_1$. Let $u_pv_p$ for some $p>1$ be the other helper edge of color $i$. Also, recall that similarly to case 2($a$), $\{u_1, u_2,\ldots,u_t\}\subset S'$ and $\{v_1, v_2,\ldots,v_t\}\subset T'$. This is because $u\in S'$ and $\mathfrak{h}_G(S,T)$ is bipartite by \Cref{lem:bip}. This means that $u_p\in S'$ and we already know that $v\in T'$. Since $u_pv_p$ is an edge of color $i$, $u_p$ is adjacent to $V(C_i)$ in $G$. We already know that $v$ is adjacent to $V(C_i)$ in $G$. Therefore, there exists a helper edge $u_pv$ of color $i$ in $\mathfrak{h}_G(S,T)$. Let $M_3:=\{v_2u_1, v_3u_2,\ldots, v_p u_{p-1},u_pv\}$, all these edges exist in $\mathfrak{h}_G(S,T)$ because the edges in $M_3$ other than $u_pv$ are from $M_p$ (as seen in the above paragraph) and hence are, in fact, real edges and we have already seen the existence of the edge $u_pv$ which is a helper edge of color $i$. Let $M_4:=\{u_1v_1,u_2v_2,\ldots,u_pv_p\}\subset M_1$. Let $M_5:=(M_1\setminus M_4)\cup M_3$. The perfect matching $M'=(M\setminus M_1)\cup M_5$ contains a helper edge $u_pv$ of color $i$ (adjacent to $v$) and has a greater intersection (in size) with $M_p$ compared to $M$, which is a contradiction. Therefore, the guards on $S$ can be reconfigured to $T$ while defending the attack on the edge $uv$. Thus we have shown that it is always possible for the guards to reconfigure between the vertex covers in $\mathcal{F}$ while defending every attack, the graph $G$ is Spartan.
\end{proof}
\section{Concluding Remarks} \label{sec:conclusions}
In this paper, we give a necessary and sufficient condition for a graph $G$ to be Spartan, i.e., to satisfy $\mathsf{evc}(G)=\mathsf{mvc}(G)$. However, there are several directions to be pursued further. An important question is whether the complexity of checking whether a given graph $G$ is Spartan is less than that of computing $\mathsf{evc}(G)$. In terms of checking whether $\mathsf{evc}(G)=\mathsf{mvc}(G)$, although we have a complete characterization, the question of finding a simpler characterization remains open. In particular, it would be interesting to know whether there exists a graph $G$ such that every vertex of $G$ belongs to a strongly (weakly) good vertex cover, but $\mathsf{evc}(G)\neq \mathsf{mvc}(G)$. If not, it would be good to know the proof that this condition is indeed sufficient to guarantee $\mathsf{evc}(G)=\mathsf{mvc}(G)$. Also, we know that there exist weakly good vertex covers that are not strongly good and hence can be destroyed by the attacker (hence, not $\mathsf{evc}$ configurations). However, we still do not know whether the condition: ``For each vertex $v$ there exists a weakly good minimum vertex cover containing $v$'' implies the following condition: ``For each vertex $v$ there exists a strongly good minimum vertex cover containing $v$''.

\bibliography{refs}

\end{document}